\documentclass[aps,pra,onecolumn,superscriptaddress,floatfix,nofootinbib,showpacs,longbibliography]{revtex4-2}
\usepackage[utf8]{inputenc}  
\usepackage[T1]{fontenc}     
\usepackage[british]{babel}  
\usepackage[sc,osf]{mathpazo}
\usepackage[scaled=0.86]{berasans} 
\usepackage[colorlinks=true, citecolor=blue, urlcolor=blue]{hyperref}  
\usepackage{graphicx} 
\usepackage[babel]{microtype}  
\usepackage{amsmath,amssymb,amsthm,bm,amsfonts,mathrsfs,bbm} 

\usepackage{xspace}  
\usepackage{pgf,tikz}
\usepackage{xcolor}
\usepackage{multirow}
\usepackage{array}
\usepackage{bigstrut}
\usepackage{braket}
\usepackage{color}
\usepackage{natbib}
\usepackage{multirow}
\usepackage{mathtools}
\usepackage[normalem]{ulem}
\usepackage{float}
\usepackage[caption = false]{subfig}
\usepackage{xcolor,colortbl}
\usepackage{color}

\newcommand{\be}{\begin{equation}}
\newcommand{\ee}{\end{equation}}
\newcommand{\ba}{\begin{eqnarray}}
\newcommand{\ea}{\end{eqnarray}}

\newtheorem{theorem}{Theorem}
\newtheorem{corollary}{Corollary}
\newtheorem{definition}{Definition}
\newtheorem{proposition}{Proposition}
\newtheorem{observation}{Observation}

\newtheorem{remark}{Remark}
\newtheorem{lemma}{Lemma}
\usepackage{mathtools}
\DeclarePairedDelimiter\ceil{\lceil}{\rceil}

\begin{document}

\title{Quantum theory is exclusive: a distributed computing setup}
\author{Sutapa Saha}
\email{sutapa.gate@gmail.com}    
\affiliation{Department of Astrophysics and High Energy Physics, S. N. Bose National Centre for Basic Sciences, JD Block, Sector-III, Salt Lake City, Kolkata - 700 106 India.}

\author{Tamal Guha}
\email{g.tamal91@gmail.com}    
\affiliation{Department of Computer Science, The University of Hong Kong, Pokfulam Road, Hong Kong.}
\author{Some Sankar Bhattacharya}
\email{somesankar@gmail.com} 
\affiliation{International Centre for Theory of Quantum Technologies, University of Gdansk, Wita Stwosza 63, 80-308 Gdansk, Poland.}

\author{Manik Banik}
\email{manik11ju@gmail.com}    
\affiliation{Department of Physics of Complex Systems, S. N. Bose National Center for Basic Sciences, Block JD, Sector III, Salt Lake, Kolkata 700106, India.}

\begin{abstract}

The framework of distributed computing, consisting of several spatially separated input-output servers, has immense importance in distant data manipulation. One of the most challenging parts of this setting is to optimize the use of information transmission lines among distant servers. In this work, we have modeled such a physically motivated distributed computing setup for which quantum communication outperforms its classical counterpart, in terms of a limited usage of noiseless transmission lines. Moreover, a broader class of communication entities, that allow state-effect description more exotic than quantum and are described within the framework of generalized probabilistic theory, also fail to meet the strength of quantum theory. The computational strength of quantum communication has further been justified in terms of a stronger version of this task, namely the delayed-choice distributed computation. The proposed task thus provides a new approach to operationally single out quantum theory in the theory-space and hence promises a novel perspective towards the axiomatic derivation of Hilbert space quantum mechanics.
\end{abstract}

\maketitle

\section{INTRODUCTION}

Computation is one of the most profound achievements of human scientific endeavour that shapes the modern era. A comprehensive understanding of its naive foundations demands interdisciplinary study on mathematics and logic \cite{Copeland04}, computer science \cite{Knuth}, cognitive sciences \cite{McCorduck04}, and physics \cite{Landauer61,Bennett73,Bennett82,Lloyd00,Aranoson04}. Mathematically, a computation can be represented as a function from some input string to output string. 
In a physical model of computation, the input  strings are encoded in a physical system on which, depending upon the computable function, some physical processes are performed to obtain the desired output value. By its name, the distributed computation suggests a highly complex, however practically relevant, computing scenario, where the input strings are distributed among multiple non-communicating spatially separated servers and the outputs are, in general, computed at the output servers -- some distant computers. Every individual input servers hence are allowed to transmit their received data to the distant computers via information transmission lines. This scenario mimics the framework of distant data manipulation, which challenges modern technology to reduce the resource requirement for transferring the data from individual input servers. For instance, when a massive celestial object is observed by multiple telescopes at different geographic locations on earth and their observed data is finally computed at a distant computing lab, one needs to access a huge number of perfect transmission lines between those input servers and the computing lab.
Although the inputs and outputs in this distributed setup are considered to be the strings of classical bits, but while the  question of transferring those information among different servers arises, they can be encoded in the states of different systems described by different operational theories, viz., classical, quantum or even more exotic systems than quantum \cite{Popescu94, Barrett07, Chiribella11, DallArno17}. The transmission lines between the input and output servers should be then chosen accordingly to be compatible with the concerning theories and finally the computation is accomplished by performing a suitably chosen measurement on the encoded systems.

Depending upon the configuration of these servers in spacetime, the encoding and decoding steps can be of two types - global and local. Global implementation of encoding and decoding requires all the servers to be at same spacetime point so that any joint physical process can be performed for the required computation, whereas in local case the servers are spatially separated and accordingly their actions are limited. This results in four broad classes of computational scenarios -- (i) local-local, (ii) global-local, (iii) global-global and (iv) local-global; the first type characterizes the encoding procedure and the second stands for decoding procedure. For instance, local discrimination of multipartite product states can be considered as local-local computational scenario where classical information encoded in product states needs to be decoded locally. The phenomenon of 'nonlocality without entanglement' studied in quantum theory \cite{Bennett99,Halder19,Rout19,Rout21} as well as in generalized probability theory (GPT) \cite{Bhattacharya20} confirms instances where perfect success is not possible if the spatially separated parties are constrained to communicate classically only (restriction on the type of communication). Similarly, the local distinguishability of orthogonal entangled states \cite{Ghosh01, Walgate02, Banik21} and recently proposed local marking of such states \cite{Sen22} constitute a scenario for global-local computation, while the recently proposed 'hyper-signaling game' \cite{DallArno17} and 'pairwise distinguishability' game \cite{Naik21} can be considered as prototypes of global-global computational scenario. 
\begin{figure*}[t!]
\centering
\includegraphics[scale=0.55]{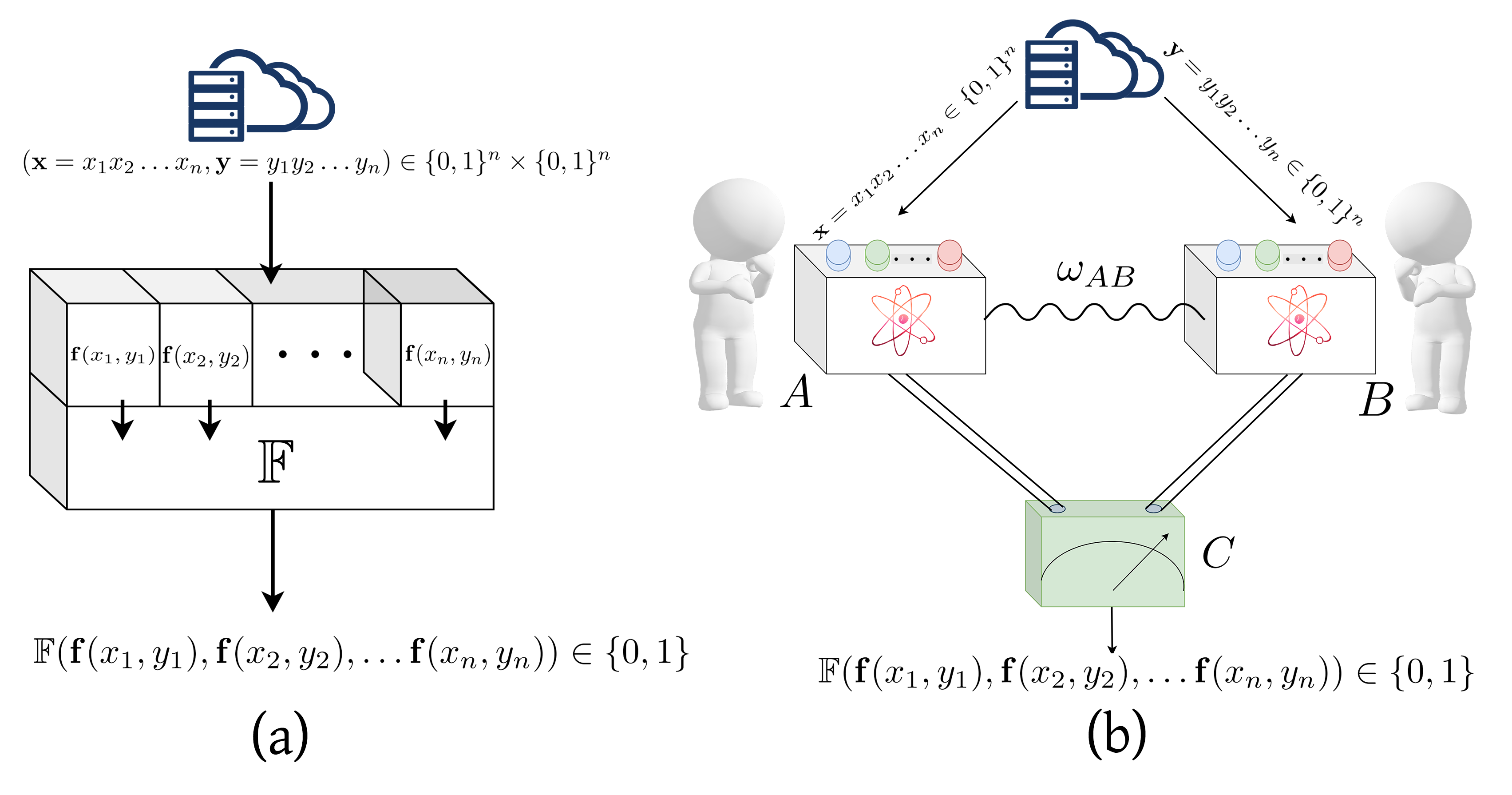}
\caption{(Color on-line) (a) {\it Dual layer computing device} receives two independent and uniformly random $n$-bit strings $x$ and $y$ from a server and outputs a bit. First, it computes the function $\mathbf{f}$ on the bits taken pairwise from $\mathbf{x}$ and $\mathbf{y}$ and finally computes $\mathbb{F}$ on the outputs of the first layer, {\it i.e.}, the dual layer computation can be represented as $(\mathbb{F},\mathbf{f}):\{0,1\}^n\times\{0,1\}^n\mapsto\{0,1\}$. (b) Corresponding {\it distributed computing scenario}: Non-communicating labs $A$ and $B$ receive two $n$-bit strings $\mathbf{x}$ and $\mathbf{y}$ respectively from a server. $A$ and $B$ are allowed to encode their strings' information in the state of systems $S_A$ and $S_B$ that can initially be prepared in some correlated state $\omega_{AB}$. The computer $C$ needs to perform a measurement on systems received from $A$ and $B$ to simulate the output of the dual layer computing device. In delayed choice version of this task, the function $\mathbb{F}$ is declared at a later stage after the communications from $A$ to $C$ and from $B$ to $C$ have taken place.}\label{fig1}
\end{figure*}
In the present work we propose a computational scenario that stands as an appropriate example of the fourth type, {\it i.e.} the local-global computational scenario. The marginal constituents of the encoded system possessed by each receiving server will be constrained by their type and by their information carrying capacity which motivates us to call this computing scenario {\it distributed computing with limited communication} (DCLC). Quite interestingly, we find that there exist DCLC tasks, exactly computable in quantum theory, does not allow perfect computation in classical theory as well as several other GPTs allowing more exotic state or/and effect space structures than quantum theory. We also provide a characterization of such tasks that can be done perfectly in quantum theory. We then study a variant of DCLC task where part of the computing function will be known after the communication from the servers to the computer is completed. We call this variant {\it delayed choice}-DCLC, {\it i.e.} DCDCLC and in short denote it as  DC$^2$LC. Interestingly, perfect accomplishment of certain  DC$^2$LC depends on the structure of the operational theory considered to model the communicating systems. The present work therefore initiates a novel approach in identifying quantum theory as an island in the theory-space. At this point it should be noted that imposing restriction on the type or the capacity of the individual transmission lines is crucial to obtain distinction among different theories. Without any such restriction any distributed computation can always be performed perfectly since input and output strings are of classical bits.\\\\

\section{PRELIMINARIES}
In this section we will first introduce the operational framework of \textit{Generalized Probabilistic Theories} (GPTs) under which we will define the paradigm of \textit{Distributed Computation with Limited Communication} (DCLC) task.
\subsection{Generalized Probabilistic Theories}\label{2.1}
Origin of this framework dates back to 1960's \cite{Mackey63,Ludwig67,Mielnik68}, and it has gained renewed interest in the recent past \cite{Barrett07,Chiribella11,Hardy01}. A GPT is specified by a list of system types and the composition rules specifying combination of several systems. 
The operational framework of GPT consists of three primary structural ingredients: \textit{Preparation, Transformation} and \textit{Measurement}, along with a composition rule for their multipartite versions. The outcome of a preparation device is termed as a \textit{state} $\omega$ of the concerned system, which specifies the outcome probabilities for all the measurements allowed to be performed on it. For a given system, the set of all possible states forms a compact and convex set $\Omega$ embedded in a positive convex cone $V_+$ of some real vector space $V$. Convexity of $\Omega$ assures that any statistical mixture of states is a valid state. The extremal points of $\Omega$, that do not allow any convex decomposition in terms of other states, are called pure states or states of maximal knowledge. 

We also identify a linear functional ($e$) over $V$ as a valid effect if it assigns probability value on every possible $\omega\in\Omega$, i.e., $0\leq e(\omega)\leq1,~\forall\omega\in\Omega$.
The set of effects $\mathcal{E}$ is embedded in the positive dual cone $V^\star_+$.  
Now, a $d$-outcome measurement is specified by a collection of $d$ effects, $M\equiv\{e_j~|~\sum_{j=1}^de_j=u\}$, where $u$ is a unique effect in $V^\star_+$, such that $u(\omega)=1,~\forall\omega\in\Omega$. With the complete characterization of allowed measurements in a GPT framework, one can introduce the idea of distinguishable states from an operational perspective which consequently leads to the concept of \textit{Operational dimension}.
\begin{definition}\label{d1}
	Operational dimension of a system $(S)$ is the largest cardinality of the subset of states, $\{\omega_i\}_{i=1}^n\subset\Omega$, that can be perfectly distinguished by a single measurement, {\it i.e.}, there exists a measurement, $M\equiv\{e_j~|~\sum_{j=1}^ne_j=u\}$, such that, $e_j(\omega_i)=\delta_{ij}$. 
\end{definition}
The operational dimension is generally different from the dimension of the vector space $V$ in which the state space $\Omega$ is embedded. For instance, in case of qubit state space, the set of density operators $\mathcal{D}(\mathbb{C}^2)$ acting on $\mathbb{C}^2$ is embedded in $\mathbb{R}^3$. However, the operational dimension of this system is $2$, as at most two qubit state can be perfectly distinguished by a single measurement -- {\it e.g.} $\{\ket{0},\ket{1}\}$ states can be perfectly distinguished by Pauli measurement along $z$-direction.

Another important ingredient to complete the operational framework for GPT is the transformations, which characterizes the set of possible time evolution on the set of states \cite{Barrett07}. In the present paper, however, we are restricted on the set of reversible transformations only, which are regarded as the symmetries on the state space: $\mathcal{T}(\Omega)\mapsto\Omega$. Note that, these transformations are linear to preserve statistical ignorance and normalization preserving to be a deterministic one. 

With the complete operational description of an individual GPT system, the natural question arises regarding their multipartite compositions. Composite systems of a GPT must be constructed in accordance with no signaling (NS) principle that prohibits instantaneous communication between two distant locations. This, along with the assumption of {\it tomographic locality} \cite{Hardy13}, assures that the composite state space lies in between two extremes - the maximal and the minimal tensor product \cite{Namioka1969}. 
\begin{definition}
	The maximal tensor product, $\Omega^A\otimes_{\max}\Omega^B$, is the set of all bi-linear functionals, $\phi: (V^A)^\star\otimes (V^B)^\star\mapsto\mathbb{R}$, such that, (i) $\phi(e^A,e^B)\ge 0$, for all $e^A\in\mathcal{E}^A$ and $e^B\in\mathcal{E}^B$, and (ii) $\phi(u^A, u^B) = 1$, where $u^A$ and $u^B$ are unit effects for system $A$ and $B$ respectively. 
\end{definition}
\begin{definition}
	The minimal tensor product, $\Omega^A\otimes_{\min}\Omega^B$, is the convex hull  of the product states $\omega^{A\otimes B}(=\omega^A\otimes \omega^B)$.
\end{definition}
States belonging in $\Omega^A\otimes_{\min}\Omega^B$ are called separable; otherwise, they are entangled. Under the assumption of no restriction \cite{Chiribella11, Janotta2013} hypothesis one can choose the corresponding effect spaces as

$\mathcal{E}^A\otimes_{\max(\min)}\mathcal{E}^B:=\{e^{AB}\in(V^A)^*\otimes(V^B)^*|0\leq e^{AB}(\omega^{AB})\leq1,~\forall\omega^{AB}\in\Omega^A\otimes_{\min(\max)}\Omega^B\}$. 

It is important to note that, the standard composite structure quantum systems is neither the minimal nor the maximal; it lies strictly in between. As a result, it contains both the entangled states as well as entangled effects in its state and effect space respectively.

In Appendix \ref{6.1} we have  discussed a particular class of such GPTs, namely the \textit{Polygon models}, relevant in the context of paper. The importance of choosing such GPTs in comparison to quantum theory is two-fold. A continuous extension of such state space structures leads to the circular cross-section of qubit Bloch ball, as a consequence they exhibits several non-classical behaviors similar to quantum theory \cite{}. On the other hand, the bipartite correlations obtained from these theories completes the 2-2-2 no-signalling polytope \cite{Barrett07}.


\subsection{Distributed computing scenario}
In the above mentioned operational framework for GPTs, here we will introduce the task of distributed computation.\par 
The scenario consists of dual layer functions $\mathcal{F}\equiv(\mathbb{F},\mathbf{f}):\{0,1\}^{2n}\to\{0,1\}$ defined as follows. The input strings $\mathbf{x}\equiv x_{1}\cdots x_{n}\in\{0,1\}^n$ and $\mathbf{y}\equiv y_{1}\cdots y_{n}\in\{0,1\}^n$, sampled independently and randomly from a cloud server, are distributed between two non communicating servers $A$ and $B$ personified as Alice and Bob, respectively. Alice and Bob encodes their respective inputs in the state of their respective systems (denoted as $S_A$ and $S_B$), and send the systems to a distant computer $C$, personified as Charlie, where the final computation $\mathcal{F}(\mathbf{x},\mathbf{y})\equiv\mathbb{F}(z_{1},\cdots,z_{n})$ with $z_{i}=\mathbf{f}(x_{i},y_{i})$ takes place and Charlie produces a single bit output (see Fig.\ref{fig1}). Here and henceforward, we will use the notation DCLC($n$) to indicate that the inputs $\mathbf{x}$ and $\mathbf{y}$ are $n$-bit strings. A more general structure of the aforementioned dual-layer computation can be proposed where $z_{i}=\mathbf{f}_{i}(x_{i},y_{i})$ and $\mathbf{f}_{i}$'s are different functions for different $i\in\{1,2,\cdots,n\}$. However, in this work we will be restricted to the case where all $\mathbf{f}_i$'s are identical. On the other hand, restriction on the information carrying capacities of $S_{A}$ and $S_{B}$ make this computational scenario nontrivial. For instance, not all DCLC($n$) tasks can be done perfectly if only $(n-1)$-cbits are allowed from each of Alice and Bob to Charlie. For an arbitrary theory, the limitation on communication can be put through some operational means. One such way is to restrict their \textit{operational dimension} (see Definition \ref{d1}) to be 2, i.e., equivalent to a cbit. Restriction on communication might also be imposed from other information theoretic motivations \cite{DallArno17,Frenkel15,Tavakoli19}.\par
Our DCLC task can be seen as a close cousin of the well known \textit{simultaneous message passing model} \cite{Yao79}, where two distant parties Alice and Bob are supposed to communicate their local bit strings ($x$ and $y$) to a distant referee, who has to compute a binary function $f(x,y)$. However, the present task of DCLC($n$) considers a dual-layered quantum computation $F(f(x),f(y))$. While an one-to-one correspondence between such dual layered computations and communication complexity \cite{Hromkovic97, Buhrman10} in the asymptotic cases has already been reported in \cite{Buhrman98}, the present scenario considers their perfect accomplishments in the single-shot regime. Moreover, here the communicating systems are not only restricted in classical or quantum theory: It allows communication of the elementary systems of a broader class of GPTs, with exotic state and effect space descriptions.

Note that, the present scenario of DCLC assumes that all the parties (Alice, Bob, and Charlie) know both the functions $\mathbb{F}$ \& $\mathbf{f}$ and choose their encoding and decoding strategies accordingly. An interesting variant of the task can be introduced where part of the computing function remains oblivious to Alice and Bob prior to their communication(s) to Charlie. More particularly, the function $\mathbf{f}$ is known to all apriori, but Alice and Bob learn about the function $\mathbb{F}$ only after they communicate to Charlie -- denoted henceforth as DC$^2$LC task. Interestingly, we establish that perfect accomplishment of some DC$^2$LC task demands specific structures in the state and effect spaces of the operational theories. Depending upon whether the systems $S_A$ and $S_B$ are taken to be classical or quantum or the elements of post-quantum GPT, the strategies executing a DCLC/ DC$^2$LC are respectively called classical, quantum, and post-quantum strategies.

\subsection {Trivial computation}
In the aforesaid distributed scenario, there exists some functions which can be trivially computed and the rest of the functions can not have such trivial strategy. Such trivial computation is formally defined as follows. 
\begin{definition}\label{def1}
A dual layer computation $(\mathbb{F},\mathbf{f})\in$DCLC($n$) is said to be trivial whenever there exists a perfect classical strategy that involves no more than $(n-1)$-cbit from each of the servers to the computer; otherwise the computation is said to be nontrivial.
\end{definition}
Note that the OD for $(n-1)$-cbit is $2^{(n-1)}$, i.e., the total numbers of bit strings, which can be perfectly discriminated. In a similar spirit, for an arbitrary GPT, we restrict each of Alice and Bob to communicate a system with OD $\leq2^{(n-1)}$ to the distant computer, in a DCLC($n$) settings.

Evidently, there are total $2^{2^n}\times2^{2^2}$ number of different DCLC($n$) tasks for a given $n$, among which some of them are trivial and others nontrivial. For instance, a computation $(\mathbb{F},\mathbf{f})$ is trivial whenever one of the functions is a constant function. Importantly, there exist trivial computations where neither $\mathbb{F}$ nor $\mathbf{f}$ is a constant. One such example is $(\mathbb{F}\equiv\oplus,\mathbf{f}\equiv\oplus)$, where $`\oplus'$ denotes the logical exclusive disjunction (XOR) operation. Triviality follows from the fact that $\oplus_{i=1}^nz_{i}=\oplus_{i=1}^n(x_{i}\oplus y_{i})=(\oplus_{i=1}^nx_{i})\oplus(\oplus_{i=1}^ny_{i})$, {\it i.e.} Charlie can do the required computation if Alice and Bob inform parity of their respective strings which requires only $1$-cbit communication from each of the transmitters to Charlie. 

\section{RESULTS}
We start this section by characterizing the set of the trivial distributed dual-layered computations, which can be computed by communicating classical systems to Charlie and hence can be perfectly accomplished by any arbitrary broader class of GPTs. Consequently, we will consider the status of nontrivial DCLC tasks in the framework of GPTs (including Quantum Theory) and conclude about the uniqueness of bipartite quantum systems as an island in the theory-space. 
\subsection{Trivial computations and their computability}
In the following, we will identify the dual-layered functions ($\mathbb{F}$,$f$), only which can be computed trivially DCLC($2$) settings. Moreover, these pairs of functions form a particular class of trivial computations even for DCLC($n$) scenario.
\begin{proposition}\label{prop1}
A dual layer computation $(\mathbb{F},\mathbf{f})\in\mbox{DCLC}(2)$ is trivial if and only if any one of the following criteria is satisfied: 
\begin{itemize}
\item[(i)] at least one of the two functions is a constant function;
\item[(ii)] at least one of them is a single-bit function; 
\item[(iii)] $\mathbb{F}$ is symmetric on inputs and $\mathbf{f}$ can be realized through $\mathbb{F}$ [and with single-bit NOT operation], {\it i.e.} $\mathbf{f}(a_1,a_2)=\mathbb{F}(a_1,a_2)$ [$\mathbf{f}(a_1,a_2)=\mathbb{F}(\bar{a}_1,\bar{a}_2)$].
\end{itemize}
\end{proposition}

\begin{proof}
  {\bf (i)} If $\mathbb{F}$ is a constant function, Charlie can yield the required output which is independent of the bit strings received from the transmitters. On the other hand, when $\mathbf{f}$ is constant, a single input pair $(z_{1},z_{2})$, with $z_1=z_2$, will be fed into the computer $C$, effectively. In both of these cases, the dual layer computation $(\mathbb{F},\mathbf{f})$ can be done perfectly even without any communication from the transmitters to the computer.

   {\bf (ii)} We will call a function $\mathbb{G}:\{0,1\}^n\to\{0,1\}$ to be single-bit function if $\forall\mathbf{a}\in\{0,1\}^n$ the functional value $\mathbb{G}(\mathbf{a})$ only depends on a single bit $a_i$ for some fixed $i\in\{1,\cdots,n\}$. 

If $\mathbb{F}$ is a single bit function, Charlie only requires information about one of $z_{1}$ and $z_{2}$ to execute the dual layer computation. The transmitters will accordingly send the corresponding bit of their strings. If $\mathbf{f}$ is a single bit function, the DCLC $(\mathbb{F},\mathbf{f})$ will effectively depend on one of the input strings - $\mathbf{x}$ or $\mathbf{y}$. Now, Alice or Bob will perform the required computation accordingly and communicate the $1$-bit output to Charlie.

    {\bf (iii)} A function will be called symmetric if it is of the form either $\mathbb{G}(\mathbf{a})=a_1\star a_2\star\cdots\star a_n$ or $\mathbb{G}(\mathbf{a})=\bar{a}_1\star \bar{a}_2\star\cdots\star \bar{a}_n$ for some binary operation $\star$. Otherwise, it is called non-symmetric.

    Let $\mathbf{f}$ (denoted by $\star$) be realized by $\mathbb{F}$ (denoted by $\circ$) and single bit NOT operation. Since, we consider $\mathbb{F}$ as symmetric, therefore $\mathbf{f}(\alpha,\beta)=\alpha\star\beta=\bar{\alpha}\circ\bar{\beta}=\mathbb{F}(\bar{\alpha},\bar{\beta})$, where $\alpha,\beta\in\{0,1\}$. In the dual-layer computation, we have,
\begin{eqnarray}
	\nonumber\mathbb{F}(z_1,z_2)&=&z_{1}\circ z_{2}
	=\mathbf{f}(x_1,y_1)\circ \mathbf{f}(x_2,y_2)\\
	\nonumber&=&(x_{1}\star y_{1})\circ(x_{2}\star y_{2})
	=(\bar{x}_{1}\circ\bar{y}_{1})\circ(\bar{x}_{2}\circ\bar{y}_{2})\\
	\nonumber&=&(\bar{x}_{1}\circ\bar{x}_{2})\circ(\bar{y}_{1}\circ\bar{y}_{2}).
\end{eqnarray}
Alice and Bob thus compute a single bit data from their respective inputs and send it to Charlie. Same holds true if $\mathbf{f}(\alpha,\beta)=\mathbb{F}(\alpha,\beta)$. 
\end{proof}

As it turns out, out of $256$ computations $176$ are trivial and the rest are nontrivial. Furthermore, among the trivial computations $60$ can be accomplished even without any communication from Alice \& Bob to Charlie, $56$ require communication from only one of Alice and Bob to Charlie, and the rest require communications from both. The proof technique used for Proposition \ref{prop1} leads us to the following generalized result.

\begin{corollary}\label{coro1}
	A dual layer computation $(\mathbb{F},\mathbf{f})\in\mbox{DCLC}(n)$, for arbitrary $n~ (\ge 2)$, is trivial if any one of the criteria in Proposition \ref{prop1} is satisfied.
\end{corollary}

\subsection{Nontrivial computations and their computability}
Referring to the operational framework introduced in Section \ref{2.1} for an arbitrary GPT and in Appendix \ref{6.1} for the specific class of polygon models, here we will briefly describe the operational model for the framework of DCLC($n$) in such theories.

To accomplish a distributed computation, Alice and Bob start their protocol with a shared bipartite state $\omega^{AB}\in \Omega^{AB}$, where $\Omega^{AB}$ is the state space for composite system with the subsystems $S_A$ \& $S_B$ satisfying the constraints imposed on their operational dimension. Depending upon the inputs $\mathbf{x}$ and $\mathbf{y}$, Alice and Bob apply some local reversible transformations $\mathcal{T}_{\mathbf{x}}^{A}$ and $\mathcal{T}_{\mathbf{y}}^{B}$ on their respective parts of the shared system. For encoding one might consider more general local operations that are not reversible. However such operations turn out be less efficient as the set of non-classical correlations arising from bipartite entangled states generally reduces under such operations. On the other hand, application of such operations is thermodynamically costlier than reversible operations. Thus we stick to reversible encodings only. Once Charlie receives the encoded systems from Alice and Bob, he performs some decoding measurement $\mathcal{M}^{AB}\equiv\{e_{k}^{AB}~|~e_{k}^{AB}\in\mathcal{E}^{AB}~\&~\sum_ke_{k}^{AB}=u^{AB}\}$; here $\mathcal{E}^{AB}$ is the set of all bipartite effects with $u^{AB}$ denoting the unit effect. 
Post processing of the measurement outcomes completes the final computation  $\mathbb{F}(f(x_{1},y_{1}),\cdots, f(x_{n},y_{n}))$.

In the following, we will first identify the basic requirements to accomplish any of the nontrivial distributed computations in any of these GPTs.
\begin{proposition}\label{prop2}
	Any nontrivial computation $(\mathbb{F},\mathbf{f})\in\mbox{DCLC}(n)$ in a GPT necessitates presence of entanglement in bipartite state and/or effect spaces of that theory.
\end{proposition}
\begin{proof}
Recall that a DCLC is trivial (nontrivial) if it can (cannot) be perfectly accomplished by some (any) classical strategy. In the language of GPT, the state and effect spaces of a $d$-level classical system is specified by a $(d-1)$-simplex. A bipartite system, composed of two such classical systems, is described uniquely by the minimal tensor product. In other words, the composite system has unique state space, as in this case we have, $\Omega^A\otimes_{\min}\Omega^B=\Omega^A\otimes_{\max}\Omega^B$ \cite{Namioka1969}; hence, the composite system allows no entanglement neither in states nor in effects. Barker's conjecture \cite{Barker76} concerns with the converse question, {\it i.e.}, for what kind of convex sets the tensor product is unique. Recently, Aubrun {\it et al.} provide an affirmative proof to the Barker's conjecture that the minimal and maximal tensor products of two finite-dimensional proper cones coincide {\it if and only if} one of the two cones is generated by a linearly independent set, {\it i.e.}, one of the state spaces is classical \cite{Aubrun19}. The only if part of this result assures the present Proposition.        
\end{proof}

It is important to note that, the above argument is true upto the assumption of \textit{no restriction hypothesis} \cite{Chiribella11, Janotta2013}, which states that with a particular choice of state space, all possible effects, which give positive probability measure on this set, should be physically realizable. 

As a natural consequence of Proposition \ref{prop2}, the question arises whether all the nontrivial DCLC tasks can be perfectly accomplished in a GPT that allows entanglement in its state and/or effect space. In the next, we will see that this is not the case in general. To this aim, we characterize the DCLC($n$) tasks that can be perfectly accomplished in quantum theory. 

Now, recalling the previously mentioned restrictions on the OD of the communicated systems, each of Alice and Bob can communicate some quantum state $\rho\in\mathcal{D}(\mathbb{C}^d)$ to Charlie, where $d\le2^{(n-1)}$. Therefore they can start with sharing some bipartite state $\rho^{AB}\in\mathcal{D}(\mathbb{C}^d\otimes\mathbb{C}^d)$. Our next result fully classify the nontrivial computation in the simplest scenario that can be perfectly done with quantum resources.  

\begin{theorem}\label{theorem1}
A nontrivial dual layer computation $(\mathbb{F},\mathbf{f})\in\mbox{DCLC}(2)$ is perfectly computable in quantum theory {\it if and only if} $\mathbf{f}$ is a balanced function. 
\end{theorem}

\begin{proof}
Here we will prove the \textit{if} part, while the \textit{only if} part is presented in the Appendix \ref{6.2.1}.

There are exactly $^4C_2$ balanced Boolean functions $\{0,1\}^2\mapsto\{0,1\}$; out of which $4$ are single bit function and hence trivial (Proposition \ref{prop1}). The remaining two functions are XOR and X-NOR. Let us first discuss the protocol for the case $(\mathbb{F}\equiv\lor,\mathbf{f}\equiv\oplus)$. Alice and Bob start the protocol with the $2$-qubit maximally entangled state $\ket{\phi^{+}}_{AB}:=\frac{1}{\sqrt{2}}(\ket{00}_{AB}+\ket{11}_{AB})$. Depending on the inputs $\mathbf{x}$ and $\mathbf{y}$, they apply local unitary operation  $\sigma_{i}^A$ and $\sigma_{j}^B$ on their respective part of the entangled state, where $i:=2x_1+x_2$ \& $j:=2y_1+y_2$, and  $\sigma_0=\mathbb{I},~\sigma_1=\sigma_X,~\sigma_2=\sigma_X\sigma_Z$ and $\sigma_3=\sigma_Z$. Whenever $\mathbf{x}=\mathbf{y}$, Charlie receives the state $\ket{\phi^+}$, otherwise a state orthogonal to $\ket{\phi^+}$. The required computation can be exactly done by performing the $2$-outcome measurement, $\mathrm{M}_{[2]}\equiv\{\mathrm{P}_{\phi^{+}},\mathbb{I}_4-\mathrm{P}_{\phi^{+}}\}$ and declaring the outcome as $\mathrm{P}_{\phi^{+}}\to 0$ and $\lnot~\mathrm{P}_{\phi^{+}}\to 1$. Other nontrivial computations follow a similar protocol with suitable relabeling of encoding and decoding (see Table \ref{table1}).\par
Also note that there are \textit{eight} different binary functions $\mathbb{F}$ (other than those in Table \ref{table1}), for each of the $f$'s (XOR and XNOR). \textit{Two} of them are constant $[\mathbb{F}(z_1, z_2)=0 \text{ or, } 1]$,  \textit{four} of them are single bit functions $[\mathbb{F}(z_1, z_2)=z_1 \text{ or, } \bar{z}_1 \text{ or, } z_2 \text{ or, } \bar{z}_2]$ and other \textit{two} are XOR and XNOR respectively. Evidently, these functions
satisfy the conditions (i), (ii) and (iii) respectively, listed in Proposition \ref{prop1}. Therefore, they are computable with $1$ cbit communication from each Alice and Bob to Charlie and then obviously with $1$ qubit communication.
\end{proof}
     \begin{center}
		\begin{table}[h!]
			\begin{tabular}{ | m{4cm}| m{4cm} | m{3cm}| m{2cm} | } 
				\hline
				$~~~~~~\mathbf{f}\equiv$ {\bf XNOR}& $~~~~~~\mathbf{f}\equiv$ {\bf XOR}& $~~~~~~~~~\sigma_{j}^{B}$ & {\bf Outcome}\\
				\hline\hline
				$~~\mathbb{F}(z_1,z_2)\equiv (\overline{z_{1}\land z_{2}})$& $~~\mathbb{F}(z_1,z_2)\equiv z_{1}\lor z_{2}$ & $~~~j=2y_{1}+y_{2}$ & $\mathrm{P}_{\phi^+}\to0$ \\
				\hline
				$~~\mathbb{F}(z_1,z_2)\equiv z_{1}\land z_{2}$& $~~\mathbb{F}(z_1,z_2)\equiv (\overline{z_{1}\lor z_{2}})$ & $~~~j=2y_{1}+y_{2}$ & $\mathrm{P}_{\phi^+}\to1$ \\
				\hline		
				$~~\mathbb{F}(z_1,z_2)\equiv (\overline{z_{1}\lor z_{2}})$& $~~\mathbb{F}(z_1,z_2)\equiv z_{1}\land z_{2}$ & $~~~j=2\bar{y}_{1}+\bar{y}_{2}$ & $\mathrm{P}_{\phi^+}\to1$ \\
				\hline		
				$~~\mathbb{F}(z_1,z_2)\equiv z_{1}\lor z_{2}$& $~~\mathbb{F}(z_1,z_2)\equiv (\overline{z_{1}\land z_{2}})$ & $~~~j=2\bar{y}_{1}+\bar{y}_{2}$ & $\mathrm{P}_{\phi^+}\to0$ \\
				\hline
				$~~\mathbb{F}(z_1,z_2)\equiv \bar{z}_{1}\land z_{2}$& $~~\mathbb{F}(z_1,z_2)\equiv z_{1}\land \bar{z}_{2}$ & $~~~j=2\bar{y}_{1}+y_{2}$ & $\mathrm{P}_{\phi^+}\to1$ \\
				\hline
				$~~\mathbb{F}(z_1,z_2)\equiv z_{1}\lor \bar{z}_{2}$& $~~\mathbb{F}(z_1,z_2)\equiv \bar{z}_{1}\lor z_{2}$ & $~~~j=2\bar{y}_{1}+y_{2}$ & $\mathrm{P}_{\phi^+}\to0$ \\
				\hline
				$~~\mathbb{F}(z_1,z_2)\equiv z_{1}\land \bar{z}_{2}$& $~~\mathbb{F}(z_1,z_2)\equiv \bar{z}_{1}\land z_{2}$ & $~~~j=2y_{1}+\bar{y}_{2}$ & $\mathrm{P}_{\phi^+}\to1$ \\
				\hline
				$~~\mathbb{F}(z_1,z_2)\equiv \bar{z}_{1}\lor z_{2}$& $~~\mathbb{F}(z_1,z_2)\equiv z_{1}\lor \bar{z}_{2}$ & $~~~j=2y_{1}+\bar{y}_{2}$ & $\mathrm{P}_{\phi^+}\to0$ \\
				\hline
			\end{tabular}
			\caption{Quantum strategies for nontrivial DCLC($2$) tasks where $\mathbf{f}$ is balanced. Alice's encoding operations $\{\sigma_{j}^A\}$ are same as used in $(\mathbb{F}\equiv\lor,\mathbf{f}\equiv\oplus)$. It is discussed in Appendix \ref{6.2} that other \textit{eight} variants of $\mathbb{F}$ are trivial.
   }\label{table1}
		\end{table}
	\end{center}

	\begin{corollary}\label{coro2}
All the computations of Theorem \ref{theorem1} are also computable in quantum theory if we consider their delayed-choice version, i.e., DC$^2$LC(2).
\end{corollary}

The protocol is discussed in detail in the Appendix \ref{6.2.2}. In this case, Alice and Bob follow the same encoding procedure(s) as of Theorem \ref{theorem1}. However, this time, Charlie performs a $4$-outcome measurement $\mathrm{M}_{[4]}\equiv\{\mathrm{P}_{\phi^+},\mathrm{P}_{\phi^-},\mathrm{P}_{\psi^+},\mathrm{P}_{\psi^-}\}$, where $\ket{\phi^\pm}:=(\ket{00}\pm\ket{11})/\sqrt{2}$ and $\ket{\psi^\pm}:=(\ket{01}\pm\ket{10})/\sqrt{2}$.
\begin{corollary}\label{coro3}
Any nontrivial dual layer computation $(\mathbb{F},\mathbf{f})\in\mbox{DCLC}(n)$ is perfectly computable, along with their delayed-choice version, i.e., \it{DC}$^2$\it{LC(n)} in quantum theory whenever $\mathbf{f}$ is a balanced function.
\end{corollary}


The proof is presented in Appendix \ref{6.2.3}, where we will show that only $\ceil{\frac{n}{2}}$-qubit communication, from each of Alice and Bob to Charlie, suffices to execute all such DCLC ($n$) computations, even with their delayed-choice versions, i.e., DC$^2$LC($n$) . This amounts to (nearly) half of the maximum allowed communication, i.e., $(n-1)$-qubit communication from each. 

We will now demonstrate the strength of bipartite compositions emerged from the broader class of GPTs, for which the state space can be visualized as regular polygons of $\mathbb{R}^2$ embedded in $\mathbb{R}^3$. In accordance with the limitation imposed on OD, all such polygon models allows only two pure preparations for perfect simultaneous discrimination. Taken in isolation, the maximum number of pair-wise distinguishable states for such theories with $n$= even, can be greater than two \cite{Brunner14} and also for those with $n$=odd, classical capacity can be greater than 1-bit \cite{Massar14}. However, interestingly for the proposed distributed computing task with limited communication, none of the extreme bipartite compositions of these GPTs can overpower a bit of classical communication, each from Alice and Bob:

\begin{theorem}\label{theorem2}
	None of the nontrivial computations in DCLC($2$) can be perfectly done in the extreme bipartite models with marginal subsystems described by symmetric polygon model. 
\end{theorem}

The complete proof is discussed in Appendix \ref{6.3}. The following observation, which characterizes a salient feature of all the nontrivial computations, lies in the core of the proof.
\begin{observation}\label{obs1}
	Consider a set $\mathcal{G}:=\{\mathbf{x},\mathbf{x}^\prime,\mathbf{y},\mathbf{y}^\prime\}$, where $\mathbf{x}\neq\mathbf{x}^\prime$ are the inputs at $A$ while $\mathbf{y}\neq\mathbf{y}^\prime$ at $B$ for a DCLC$(2)$ task. Altogether, $\binom{4}{2}\times \binom{4}{2}=36$ different such sets are possible. Evidently, the strings in $\mathcal{G}$ will be mapped into the bit values 0 and 1 respectively in the ratio $4:0$, $2:2$, $0:4$, $1:3$ and $3:1$. It turns out that at-least one $\mathcal{G}$, among the $36$ possibilities, must have the aforesaid ratio either $1:3$ or $3:1$ for every nontrivial DCLC$(2)$ task.
\end{observation}
	In the Appendix \ref{6.3}, it is then explicitly shown that none of the extreme bipartite compositions for \textit{even}, as well as \textit{odd}-gon theories can satisfy the above requirement, which prohibits them to execute any nontrivial DCLC($2$). Hence, it is needless to say that none of the DC$^2$LC(2) can be performed in such GPT models.

The strength of the pair of Theorem \ref{theorem1} and \ref{theorem2} are two-fold: Apart from the stronger communication utility of the individual polygon models (as mentioned earlier), their bipartite compositions also overpower quantum correlations both in space and time like scenario. For instance, while the maximal bipartite composition of $n=4$ polygon model exhibits strongest possible Bell-nonlocality \cite{Popescu94}, their minimal composition posses a violation in a recently introduced communication scenario, namely the Hypersignaling game \cite{DallArno17}. In contrast, for the present task of DCLC, which encapsulates an hybridization of both space-like (between the non-communicating agents Alice and Bob) and time-like (from Alice/ Bob to the referee) settings, the status of all such theories is no better than the classical one. On the other hand, the exclusiveness of quantum theory in such a settings direct towards a crucial topological ingredient, i.e., presence of entanglement in both the state and effect space. It is also important to mention that there are possibilities of various other compositions between minimal and maximal compositions of all these GPT models, which contain entangled states, as well as entangled effects \cite{DallArno17}. However, such models restrict some of the reversible transformations, due to the requirement of positive probability of all the effects over all possible preparations. In fact, the transformations are so limited that even not all the product states (effects) can be reversibly transformed among themselves. This leads us to conjecture that same non-go result of Theorem \ref{theorem2} also holds for all these intermediate models. A strong evidence in support of our conjecture has already been reported in \cite{Saha18}, considering a specific dual-layered computation $(\mathbb{F}\equiv\lor,\mathbf{f}\equiv\oplus)\in\mbox{DCLC}(2)$ for all possible intermediate compositions for $n$=4 GPT model. This, in turn, further strengthen our finding: Not only the exotic state and effect space structure of quantum theory, but the continuity in their allowed reversible transformations plays an important role to identify it as an island in theory-space, from the perspective of distributed computing. 

\section{DISCUSSION}
Apart from immense practical importance in present day distant data manipulation, our work introduces an operational task to single out the bipartite composite structure of quantum theory. Despite remarkable advantages in several computation and information protocols, from a foundational point of view it still remains illusive: \textit{why the nature is quantum}? There is no general consensus why our physical world should be modeled by Hilbert space quantum mechanics, which, from a mathematical standpoint, is just an example of model among several other possibilities \cite{Mackey63,Ludwig67}. During last two decades some novel approaches have been developed that find seemingly impossible consequences of post-quantum models and thus reject them to be the candidate theory of the physical world \cite{vanDam,Brassard06,Linden07,Pawlowski09,Navascues09,Fritz13,Aaronson04,Muller12,Pfister13,Banik19,Krumm19}. In this direction, our study proposes non-trivial example of distributive computations, where quantum theory performs better than every possible two-dimensional polygon theories, which are close cousin of a qubit both from geometric and information theoretic perspectives \cite{Janotta11}. The studies made in Refs. \cite{Barrett05,Chao17,Weilenmann20} are worth mentioning at this point. There it has been established that some $3$-party quantum correlations cannot be produced by Popescu-Rohrlich (PR) correlation \cite{Popescu94} which otherwise are possible in quantum theory. The computational limitation of the GPTs established in those works stem from the impossibility of entanglement swapping for generalized nonlocal correlations \cite{Short06}. The present work has stronger implications as we have not invoked any $3$-party correlation. Rather we  show that some $2$-party correlation outside the quantum realm is not as good as $2$-party quantum correlations. Furthermore, those works only address correlations in space-like scenario whereas the present work considers correlation in space-like as well as in time-like scenarios.

It is worth mentioning that the task of Quantum fingerprinting \cite{Buhrman01} is a special and celebrated example of our general DCLC computations. In Quantum fingerprinting Charlie is asked to calculate the function $e(\mathbf{x},\mathbf{y}):=1$ (if $\mathbf{x}=\mathbf{y}$) and $e(\mathbf{x},\mathbf{y}):=0$ (if $\mathbf{x}\neq\mathbf{y}$) using the minimum communication from Alice and Bob who are given two random $n$ bit strings $\mathbf{x}$ and $\mathbf{y}$, respectively, which can be seen as a dual layer computation  $(\mathbb{F}\equiv\lor,\mathbf{f}\equiv\oplus)$. While an exponential gap between classical and quantum resources have been established there, our result concerns with the perfect accomplishment of the task for single-shot case.
\begin{figure}[t!]
	\begin{center}
		\includegraphics[scale=0.5]{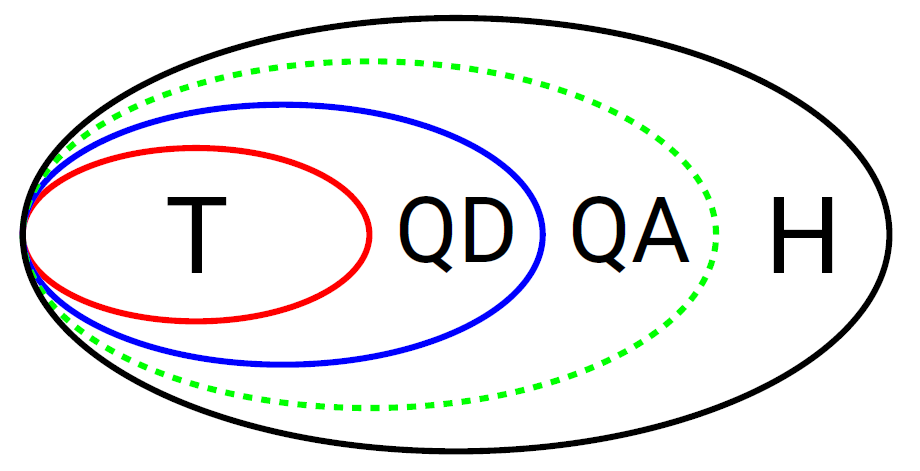}
	\end{center}
	\caption{(Color on-line) A schematic Venn diagram of different dual layer computations as seen from the perspective of quantum theory (considering DCLC(2)). The class T denotes the trivial computations (see Definition \ref{def1}). QD (in between red and blue curves) are the nontrivial computations that can be perfectly done in quantum theory (Theorem \ref{theorem1} \& Corollary \ref{coro2}) and they establish quantum advantage over classical as well as polygonal GPTs (Theorem \ref{theorem2}). QP (in between blue and black curves) represents the class of hard computations which imply probabilistic success in quantum theory [e.g. $(\mathbb{F}\equiv\oplus,\mathbf{f}\equiv\lor)$]. There might exist a subset of computations QA in QP (in between blue and dashed green curve) where quantum theory, although allowing for probabilistic success, provides advantage over polygonal GPTs. Finding an example of such a computation remains an open problem.}\label{fig2}
\end{figure}

Finally we will conclude by mentioning the open problems for single-shot distributed computing, emerged from the present work. Our result is complete in the framework of DCLC($2$) up to the trivial ones and those with perfect quantum accomplishment. However, the class of computations, for which there are probabilistic quantum advantage over the polygonal GPTs, is still unknown (see Fig \ref{fig2}). It is also instructive to completely characterize these classes for DCLC($n$) settings. Lastly, the distributed computing scenario with higher numbers ($\geq3$) of input ports can be a potential topic for further research.	

\section{ACKNOWLEDGEMENTS}
 We thank Guruprsad Kar for many stimulating discussions. We also gratefully acknowledge private communication with Markus P. Müller and thank Claude Crépeau for his useful suggestion and pointing out the relevant reference \cite{Chao17}. SSB acknowledges partial support by the Foundation for Polish Science (IRAP project, ICTQT, contract no. MAB/2018/5, co-financed by EU within Smart Growth Operational Programme). MB acknowledges research grant of INSPIRE-faculty fellowship by the Department of Science and Technology, Government of India, funding from the
National Mission in Interdisciplinary Cyber-Physical systems from the Department of Science and Technology
through the I-HUB Quantum Technology Foundation
(Grant No. I-HUB/PDF/2021-22/008), and the start-up
research grant from SERB, Department of Science and
Technology (Grant No. SRG/2021/000267).

\section{APPENDIX}
\subsection{The Structure of Polygon Theories}\label{6.1}
In the operational framework of general probabilistic theories, as discussed in Section \ref{2.1}, here we will introduce a special class of such theories. Geometrically, these models constitute with an integer $n$ numbers of pure states, situated on the vertices of a symmetric regular polygon. 

{\bf Single system:} For an elementary system, the state space $\Omega_n$ is a regular polygon with $n$ vertices. For a fixed $n$, $\Omega_n$ is the convex hull of $n$ pure states $\{\omega_i\}_{i=0}^{n-1}$, where,
\begin{equation}
	\omega_i=\begin{pmatrix}
		r_n \cos \frac{2 \pi i}{n}\\
		r_n \sin \frac{2 \pi i}{n}\\
		1\end{pmatrix},~~ \mbox{with}~~r_n=\sqrt{\sec(\pi/n)}.
\end{equation}
On the other hand, corresponding effect space $\mathcal{E}_n$, collection of all the possible measurement effects, is the convex hull of the null effect $O$, unit effect $u$, the extremal effects $\{e_j\}_{j=0}^{n-1}$, and their complementary effects $\{\bar{e}_j\}_{j=0}^{n-1}$, where, $\bar{e}_j:=u-e_j$. The null and unit effects are respectively given by $O=(0,0,0)^{\mathrm{T}}~~\text{and}~~u=(0,0,1)^{\mathrm{T}}$, where, $\mathrm{T}$ denotes the matrix transposition. 
The effects $\{e_j\}_{j=0}^{n-1}$ are given by,
\begin{center}
	\begin{tabular}{ c|c } 
		\hline
		Even-gon & Odd-gon  \\
		\hline\hline
		&\\
		$e_j=\frac{1}{2}\begin{pmatrix}
			r_n \cos \frac{(2 j-1) \pi}{n}\\
			r_n \sin\frac{(2 j-1) \pi}{n}\\
			1
		\end{pmatrix}$& ~~~~~$e_j=\frac{1}{1+r_n^2}\begin{pmatrix}
			r_n \cos \frac{2 j \pi}{n}\\
			r_n \sin\frac{2 j \pi}{n}\\
			1
		\end{pmatrix}$  \\ 
		&\\
		\hline
	\end{tabular}
\end{center}
For even-gon, it turns out that, $\bar{e}_j:=u-e_j=e_{(j \oplus_n\frac{n}{2})}$, where, $\oplus_n$ denotes {\it addition modulo $n$}. Therefore, the effects $\{e_j\}_{j=0}^{n-1}$ as well as their complementary effects are not the pure effects only, but they are the ray extremals of the effect cone $(V^\star)_+$ also. In contrast, for odd-gon, only $\{e_j:=u-e_j\}_{j=0}^{n-1}$ are the ray extremals, whereas their complementary effects $\{\bar{e}_j:=u-e_j\}_{j=0}^{n-1}$ are not despite being pure. 

For any $n$-gon theory, the set of the reversible transformations (RT), $\mathbb{T}_n$, is the dihedral group of order $2n$ containing $n$ reflections and $n$ rotations, {\it i.e.},
\begin{align}
	\mathbb{T}_n\equiv\left\{\mathcal{T}_k^p~|~k=0,\cdots,n-1;~\&~p\in\{+,-\} \right\},\nonumber\\
	\mathcal{T}_k^p:= 
	\begin{pmatrix}
		\cos\frac{2\pi k}{n} & -p\sin\frac{2\pi k}{n} & 0 \\
		\sin\frac{2\pi k}{n} & p\cos\frac{2\pi k}{n} & 0 \\
		0  & 0  & 1   
	\end{pmatrix},
\end{align}
with $p=+$ corresponds to the rotation and $p=-$ to the reflection.

{\bf Bipartite system:} Any bipartite composition of $n$-gon systems must include $n^2$ factorized states,
\begin{eqnarray}
	\Omega^{product}:=\left\{\omega_{ni+j}^{A\otimes B}:=\omega_i^A\otimes\omega_j^B~|~i,j\in\{0,\cdots,n-1\}\right\}\subset\Omega_{n^{\otimes2}}:=\Omega^{AB}_n.
\end{eqnarray}
We will use the superscript $A\otimes B$ to denote factorizability. For the bipartite system, the product effects are of the form,
\begin{eqnarray}
	\mathcal{E}^{product}&:=&\left\{g^A\otimes g^B\right\}\subset\mathcal{E}^{AB}_n,\\
	\mbox{where,}&&~~g^X\in\left\{O^X,u^X\right\}\bigcup\left\{e_i^X,~\bar{e}_i^X\right\}_{i=0}^{n-1};~~~X\in\{A,B\}.\nonumber
\end{eqnarray}
Since, $p(e^{A\otimes B}|\omega^{A\otimes B})=p(e^A|\omega^{A})p(e^{B}|\omega^{B})$, therefore,
\begin{equation}
	0\leq p(e^{A\otimes B}|\omega^{A\otimes B})\le 1;~~\forall~e^{A\otimes B}\in\mathcal{E}^{product}~\&~\forall~\omega^{A\otimes B}\in\Omega^{product}.  
\end{equation}

Apart from these factorized states and effects, a bipartite system may also allow non-factorized (entangled) states and effects that we will denote as $\omega^{AB}$ and $e^{AB}$ respectively. Of course, they must satisfy the consistency requirements:
\begin{eqnarray}
	0\leq p(e^{A\otimes B}|\omega^{AB})\le 1,~\forall~e^{A\otimes B}\in\mathcal{E}^{product},\\
	0\leq p(e^{AB}|\omega^{A\otimes B})\le 1,~\forall~\omega^{A\otimes B}\in\Omega^{product}.
\end{eqnarray}
In Ref.\cite{Janotta11}, the authors have introduced an maximally entangled state for bipartite $n$-gon theories both for {\it odd} and {\it even} $n$. Applying all possible local RTs $\left\{\mathcal{T}^p_k\right\}$ on Alice's part and $\left\{\mathcal{T}^q_l\right\}$ on Bob's part, we can get all the other entangled states as follows :\\
\underline{\bf Odd n:}
\begin{equation}\label{e2}
	\begin{aligned}
		\omega^{AB}_{kl}(p,q):=
		\begin{pmatrix}
			\cos \left(\frac{2\pi}{n}(k-l)\right) & -\sin \left(\frac{2\pi}{n}(k-l)\right) & 0\\
			\sin \left(\frac{2\pi}{n}(k-l)\right) & ~~~\cos \left(\frac{2\pi}{n}(k-l)\right) & 0\\
			0 & 0 & 1
		\end{pmatrix};~~\mbox{when} ~~p=q,
		&\\
		\omega^{AB}_{kl}(p,q):=
		\begin{pmatrix}
			\cos \left(\frac{2\pi}{n}(k+l)\right) & ~~~\sin \left(\frac{2\pi}{n}(k+l)\right) & 0\\
			\sin \left(\frac{2\pi}{n}(k+l)\right) & -\cos \left(\frac{2\pi}{n}(k+l)\right) & 0\\
			0 & 0 & 1
		\end{pmatrix};~~\mbox{when} ~~p\neq q.
	\end{aligned}
\end{equation}
\underline{\bf Even n:}
\begin{equation}\label{e3}
	\begin{aligned}
		\omega^{AB}_{kl}(p,q):=
		\begin{pmatrix}
			\cos \left(\frac{2\pi}{n}(k-l)-p\frac{\pi}{n}\right) & -\sin \left(\frac{2\pi}{n}(k-l)-p\frac{\pi}{n}\right) & 0\\
			\sin \left(\frac{2\pi}{n}(k-l)-p\frac{\pi}{n}\right) & ~~~\cos \left(\frac{2\pi}{n}(k-l)-p\frac{\pi}{n}\right) & 0\\
			0 & 0 & 1
		\end{pmatrix};~~\mbox{when} ~~p=q,
		&\\
		\omega^{AB}_{kl}(p,q):=
		\begin{pmatrix}
			\cos \left(\frac{2\pi}{n}(k+l)-p\frac{\pi}{n}\right) & ~~~\sin \left(\frac{2\pi}{n}(k+l)-p\frac{\pi}{n}\right) & 0\\
			\sin \left(\frac{2\pi}{n}(k+l)-p\frac{\pi}{n}\right) & -\cos \left(\frac{2\pi}{n}(k+l)-p\frac{\pi}{n}\right) & 0\\
			0 & 0 & 1
		\end{pmatrix};~~\mbox{when} ~~p\neq q.
	\end{aligned}
\end{equation}
Similarly, all the possible maximally entangled effects are given by,\\
\underline{\bf Odd n:}
\begin{equation}\label{e4}
	\begin{aligned}
		e^{AB}_{kl}(p,q):=\frac{1}{1+r_n^2}
		\begin{pmatrix}
			\cos \left(\frac{2\pi}{n}(k-l)\right) & -\sin \left(\frac{2\pi}{n}(k-l)\right) & 0\\
			\sin \left(\frac{2\pi}{n}(k-l)\right) & ~~~\cos \left(\frac{2\pi}{n}(k-l)\right) & 0\\
			0 & 0 & 1
		\end{pmatrix};~~\mbox{when} ~~p=q,
		&\\
		e^{AB}_{kl}(p,q):=\frac{1}{1+r_n^2}
		\begin{pmatrix}
			\cos \left(\frac{2\pi}{n}(k+l)\right) & ~~~\sin \left(\frac{2\pi}{n}(k+l)\right) & 0\\
			\sin \left(\frac{2\pi}{n}(k+l)\right) & -\cos \left(\frac{2\pi}{n}(k+l)\right) & 0\\
			0 & 0 & 1
		\end{pmatrix};~~\mbox{when} ~~p\neq q.
	\end{aligned}
\end{equation}
\underline{\bf Even n:}
\begin{equation}\label{e5}
	\begin{aligned}
		e^{AB}_{kl}(p,q):=\frac{1}{2}
		\begin{pmatrix}
			\cos \left(\frac{2\pi}{n}(k-l)-p\frac{\pi}{n}\right) & -\sin \left(\frac{2\pi}{n}(k-l)-p\frac{\pi}{n}\right) & 0\\
			\sin \left(\frac{2\pi}{n}(k-l)-p\frac{\pi}{n}\right) & ~~~\cos \left(\frac{2\pi}{n}(k-l)-p\frac{\pi}{n}\right) & 0\\
			0 & 0 & 1
		\end{pmatrix};~~\mbox{when} ~~p=q,
		&\\
		e^{AB}_{kl}(p,q):=\frac{1}{2}
		\begin{pmatrix}
			\cos \left(\frac{2\pi}{n}(k+l)-p\frac{\pi}{n}\right) & ~~~\sin \left(\frac{2\pi}{n}(k+l)-p\frac{\pi}{n}\right) & 0\\
			\sin \left(\frac{2\pi}{n}(k+l)-p\frac{\pi}{n}\right) & -\cos \left(\frac{2\pi}{n}(k+l)-p\frac{\pi}{n}\right) & 0\\
			0 & 0 & 1
		\end{pmatrix};~~\mbox{when} ~~p\neq q.
	\end{aligned}
\end{equation}
For an arbitrary $n$, the set of all possible RTs for bipartite system is given by
\begin{eqnarray}\label{tran}
	\mathbb{T}_{n^{\otimes2}}&:=&\mathbb{T}^{AB}\equiv\left\{\mathcal{S},\mathcal{T}^p_k\otimes\mathcal{T}^q_l\right\},\\
	&&k,l\in\{0,\cdots,n-1\};~~~~p,q\in\{+,-\}.\nonumber
\end{eqnarray}
$\mathcal{S}$ is the SWAP map whose action is defined as,
\begin{equation}
	\mathcal{S}(\omega^A\otimes\omega^B)=\omega^B\otimes\omega^A;~~\forall~\omega^A\in\Omega^A~\&~\omega^B\in\Omega^B.
\end{equation}
\begin{remark}\label{rem4}
	The SWAP map is a global transformation, {\it i.e.}, it cannot be implemented locally. Alternatively, any local transformation $\mathcal{T}\in\mathbb{T}^{AB}$ never maps a product state (effect) to an entangled one and vice versa. In other words, the set of product states (effects) and the set of entangled states (effects) are two disconnected islands under the reversible transformation $\mathbb{T}^{AB}$.
\end{remark}
Positivity of the predicted probabilities imposes restrictions on the states, effects and transformations that can be allowed together in a composite system. Satisfying this consistency requirement, several composite models are possible. These models can be classified into three main types as discussed below. 

\textbf{Type-I:} {\it Entangled states product effects model} : In this case, all possible product as well as entangled states (listed above) are allowed, {\it i.e.}, $\Omega^{AB}_n$ is the convex hull of the set\\ $\left\{\omega_{ni+j}^{A\otimes B}, \omega_{kl}^{AB}(p,q)~|~i,j,k,l\in\{0,\cdots,n-1\};~p,q\in\{+,-\}\right\}$,\\ whereas the effects are only product in nature. Due to the presence of entangled states, such model can exhibit Bell {\it nonlocality} \cite{Bell64}. In fact, such a model can be stronger in {\it space-like} correlation by revealing more nonlocal behaviour than quantum theory \cite{Popescu94}.   

\textbf{Type-II:} {\it Product states entangled effects model} : It allows only the product states, {\it i.e.}, $\Omega^{AB}$ is the convex hull of the set $\left\{\omega_{ni+j}^{A\otimes B}~|~i,j\in\{0,\cdots,n-1\}\right\}$. However, it allows all possible product as well as entangled effects. Such a model is {\it local} by construction. Due to the presence of all possible entangled effects, this model can also exhibit peculiar feature. For instance, the authors have shown in Ref. \cite{DallArno17} that such a model can allow {\it time-like } correlations that are stronger than quantum theory.

\textbf{Type-III:} {\it Dynamically restricted models} : There can be some models which allow some entangled states along with some (suitably chosen) entangled effects unlike the \textbf{Type-I} and \textbf{Type-II} models. Further, due to the consistency requirement ({\it i.e.}, positivity of the outcome probability), not all the reversible transformations can be allowed when both of teh entangled states and entangled effects are incorporated; hence, it is named as `dynamically restricted models'. Such restriction prevents even all the pure states (effects) to be mapped to each other under reversible transformation which makes \textbf{Type-III} models quite uninteresting. 
\subsection{Nontrivial DCLC and their quantum computability}\label{6.2}
\subsubsection{Proof of the \textit{only if} part of Theorem 1} \label{6.2.1}
To begin with, we will first establish the following lemmas:
\begin{lemma}\label{lemma1}
	Any nontrivial dual layer computation $(\mathbb{F},\mathbf{f})\in\mbox{DCLC}(2)$ maps the set of input stings $\{\mathbf{x},\mathbf{y}\}$ into the binary bit values in $1:3$ ratio \textit{if and only if} $\mathbf{f}$ is balanced.  
\end{lemma}
\begin{proof}
	A $\{0,1\}^2\to\{0,1\}$ Boolean function produces a binary output in either of the three possible ratios -- $4:0$ (constant), $1:1$ (balanced) and $1:3$ (unbalanced). Also for a balanced function, there are only two possibilities for it: either \textit{single-bit}  or, \textit{parity} function (XOR or, XNOR). Now, if $\mathbf{f}$ is balanced and also a single-bit function, then for any $\mathbb{F}$, the dual-layered computation $(\mathbb{F},\mathbf{f})$ is trivial (see Proposition \ref{prop1}). On the other hand, if $\mathbf{f}$ is balanced but not a single bit (i.e., either XOR or, XNOR), then the pair $(\mathbb{F},\mathbf{f})$ is trivial for $\mathbb{F}$ being a constant (condition (i) in Proposition \ref{prop1}) or, a balanced (conditions (ii) and (iii) in Proposition \ref{prop1}) function.\par
    So the only possibility for nontrivial $(\mathbb{F},\mathbf{f})$, with balanced $\mathbf{f}$, is that $\mathbb{F}$ is unbalanced. Now, for balanced $f$ the bit-pair $z_1z_2$ can produce each of the values $\{00,01,10,11\}$ exactly twice. This confirms that the pair $(\mathbb{F},\mathbf{f})$, with unbalanced $\mathbb{F}$, will produce $0$ and $1$ in $1:3~(\text{or, }3:1)$ ratio and hence proves the \textit{if} part.
	
	If $\mathbf{f}$ is unbalanced and $\mathbb{F}$ is balanced, the output will be either $1:3$ or $3:5$. Alternatively, it will be among $1:15$, $3:13$ and $7:9$ for both of $\mathbf{f}$ and $\mathbb{F}$ being unbalanced. Now, consider an unbalanced $\mathbf{f}$ where the `output bit value $0:$ output bit value $1=1:3$'. In this case, bit values of $z_1z_2$ follow the ratio $00:01:10:11=1:3:3:9$. If $\mathbb{F}$ is a balanced function, such that, $\{00,01\}\to 0/1$ or $\{00,10\}\to 0/1$, output of $(\mathbb{F}$, $\mathbf{f})$ is in $1:3$ ratio. But in this case, $\mathbb{F}$ being a single-bit function makes $(\mathbb{F}$, $\mathbf{f})$ trivial (Proposition \ref{prop1}). Therefore, no other nontrivial $(\mathbb{F},\mathbf{f})$ can be in $1:3$ ratio except $\mathbf{f}$ being a balanced function.
\end{proof} 
\begin{lemma}\label{lemma2}
    Any dual layer function of DCLC($2$) with balanced final output will be a trivial computation.
\end{lemma}
\begin{proof}
Suppose that the function $f$ is neither constant, nor balanced. Then for $f(x,y)$ the 'Output bit value $0$: Output bit value $1$'= $1:3$ (or the reverse). Consequently, for such a $f$ the values of $z_{1}z_{2}$ will be $\{00,~01,~10,~11\}$ in a ratio $1:3:3:9$, identifying $z_{i}=f(x_{i}, y_{i})$. Now, observe that no $\mathbb{F}$ can be defined on the arguments $\{z_{1},z_{2}\}$, which will produce the final output $0$ and $1$ in $1:1$ ratio. Therefore, for final output to be $1:1$, either $f$ is constant or, balanced. If $f$ is constant, then from Proposition. 1 the computation $(\mathbb{F},f)$ is trivial. On the other hand, if $f$ is balanced then the pair $\{z_{1},z_{2}\}$ takes the values $00:01:10:11$ uniformly. Then the function $\mathbb{F}$ should also be either a single-bit (which is again trivial from Proposition. 1), or balanced function to generate the final output in $1:1$ ratio. Hence, if both of the function $\mathbb{F}$ and $f$ are balanced but not single bit, then they can be XOR or, XNOR and in both cases they are trivial according to condition (iii) of Proposition. 1. This finally helps us to conclude that the final output of the pair $(\mathbb{F},f)$ is $1:1$ only if the computation is trivial.
\end{proof}

Lemma \ref{lemma1} and \ref{lemma2} assures that to prove the {\it only if} part of Theorem \ref{theorem1}, it is sufficient to prove that no quantum strategy (entangled states along with local unitaries and two outcome measurements) can produce two disjoint subspaces containing states other than $1:3~\text{or}~1:1$. Consider that Alice and Bob start their protocol with a pure entangled state $\ket{\psi}=a\ket{00}+b\ket{11}$, where, $\{a,b\}\in \mathbb{R}$, s.t., $a^{2}+b^{2}=1$ without loss of generality . They have some unitary encoding strategies, $\{U_{i}^{A}\}_{i=0}^3$ and $\{U_{j}^{B}\}_{j=0}^3$, respectively. Now, according to the Observation \ref{obs1}, for every nontrivial DCLC($2$) there exists at least a group of four input stings $\mathcal{G}:=\{\mathbf{x},\mathbf{x}^\prime,\mathbf{y},\mathbf{y}^\prime\}$ ($\mathbf{x}\neq\mathbf{x}^\prime$ \& $\mathbf{y}\neq\mathbf{y}^\prime$) that follows the ratio $1:3$.  Considering Alice's and Bob's encoding as $U_{0}^{A}, U_{1}^{A}$ and $U_{0}^{B}, U_{1}^{B}$, the resulting encoded states read 
\begin{eqnarray}\label{e6}
	\nonumber	\mathbf{x},\mathbf{y}\longmapsto \ket{\xi_{1}}=a\ket{\psi_{0}\phi_{0}}+b\ket{\psi^\perp_{0}\phi^\perp_{0}},\\
	\nonumber		\mathbf{x}^\prime,\mathbf{y}\longmapsto \ket{\xi_{2}}=a\ket{\psi_{1}\phi_{0}}+b\ket{\psi^\perp_{1}\phi^\perp_{0}},\\
	\nonumber \mathbf{x},\mathbf{y}^\prime\longmapsto \ket{\xi_{3}}=a\ket{\psi_{0}\phi_{1}}+b\ket{\psi^\perp_{0}\phi^\perp_{1}},\\
	\nonumber \mathbf{x}^\prime,\mathbf{y}^\prime\longmapsto \ket{\xi_{4}}=a\ket{\psi_{1}\phi_{1}}+b\ket{\psi^\perp_{1}\phi^\perp_{1}},
\end{eqnarray}
where, $\ket{\psi_i}=U^A_i\ket{0}~\&~\ket{\phi_i}=U^B_i\ket{0}$, for $i\in\{0,1\}$. The orthogonality conditions $\langle\xi_{1}|\xi_{j}\rangle=0,~\forall~j\in\{2,3,4\}$, imply $\ket{\psi_{1}}=\ket{\psi^\perp_{0}}, \ket{\psi^\perp_{1}}=\pm\ket{\psi_{0}}$ and $\ket{\phi_{1}}=\ket{\phi^\perp_{0}},\ket{\phi^\perp_{1}}=\mp\ket{\phi_{0}}$. In other worlds, both $U_{0}^{A}$ \& $U_{1}^{A}$ ($U_{0}^{B}$ \& $U_{1}^{B}$) map the states $\{\ket{0},\ket{1}\}$ into same orthogonal pairs $\{\ket{\psi_0},\ket{\psi^\perp_0}\}$ ($\{\ket{\phi_0},\ket{\phi^\perp_0}\}$). For decoding, Charlie performs the measurement, $\mathrm{M}^{\xi_1}_{[2]}\equiv\{\mathrm{P}_{\xi_1},\mathbb{I}_4-\mathrm{P}_{\xi_1}\}$ and assigns the outcome as $\mathrm{P}_{\xi_1}\to 0$ and $\lnot~\mathrm{P}_{\xi_1}\to 1$. For every input string $\mathbf{x}$, either ($\mathbf{x}$, $\mathbf{y}$) yields outcome $0$, whereas ($\mathbf{x}$, $\mathbf{y}^\prime$) yields $1$ (where $\mathbf{y}\neq \mathbf{y}^\prime$ belong to Bob), or $\mathbf{x}$ forms a group like $\mathcal{G}$. In both these cases,  all the Unitaries $\{U^A_i\}_{i=0}^3$ ($\{U^B_j\}_{j=0}^3$) map the states $\{\ket{0},\ket{1}\}$ into same orthogonal pairs $\{\ket{\psi_0},\ket{\psi^\perp_0}\}$ ($\{\ket{\phi_0},\ket{\phi^\perp_0}\}$). Therefore, the encoding by Alice and Bob, without loss of any generality, can be chosen to be the Pauli matrices $\{\sigma_{i}\}_{i=0}^{3}$. Let us denote $(\sigma_{i}^{A}\otimes\sigma_{j}^{B})\ket{\psi}=\ket{\xi}_{ij}$, where, $\{i,j\}\in\{0,...,3\}$. Note that $\ket{\xi}_{ij}\sim\ket{\xi}_{kl}$ (up-to global phase) if $i+k=j+l=3$. To compute DCLC ($2$) perfectly, Charlie performs a measurement that divides the communicated bipartite states in two orthogonal subspace. Evidently, there exists only one such measurement which divides the above states in $\{\ket{00},\ket{11}\}$ and $\{\ket{01},\ket{10}\}$ subspace in $1:1$ ratio, hence achieves a trivial computation (see Lemma \ref{lemma2}). Therefore, the dual-layer computations for which the final outcome is balanced, {\it i.e.}, only the trivial ones can be performed with non-maximally pure entangled states. Whenever $a=b=\frac{1}{\sqrt{2}}$, the subspace can also be divided in $1:3$ ratio and \textit{no other} choice is possible at all.

\subsubsection{Proof of Corollary 2}\label{6.2.2}
Let us consider $\mathbf{f}$ as a balanced function in $(\mathbb{F},\mathbf{f})$, while $\mathbb{F}$ is delayed-choice, i.e., declared later, once after Alice and Bob communicate their respective bit information to Charlie. The encoding protocol here is similar to that of Theorem \ref{theorem1}. Depending upon the input strings $\mathbf{x}$ \& $\mathbf{y}$, Charlie receives the bipartite state as follows: 
\begin{center}
	\begin{tabular}{ c|c } 
		\hline
		inputs $\mathbf{x},~\mathbf{y}$ & $\sigma_{2x_1+x_2}^A\otimes\sigma_{2y_1+y_2}^B\ket{\phi^{+}}_{AB}$  \\
		\hline\hline
		$x_1=y_1~\&~x_2=y_2$& $\ket{\phi^{+}}_{AB}$  \\ 
		\hline
		$x_1\neq y_1~\&~x_2=y_2$& $\ket{\phi^{-}}_{AB}:=\frac{1}{\sqrt{2}}(\ket{00}_{AB}-\ket{11}_{AB})$  \\ 
		\hline
		$x_1=y_1~\&~x_2\neq y_2$& $\ket{\psi^{+}}_{AB}:=\frac{1}{\sqrt{2}}(\ket{01}_{AB}+\ket{10}_{AB})$  \\ 
		\hline
		$x_1\neq y_1~\&~x_2\neq y_2$& $\ket{\psi^{-}}_{AB}:=\frac{1}{\sqrt{2}}(\ket{01}_{AB}-\ket{10}_{AB})$  \\ 
		\hline
	\end{tabular}
\end{center}
For decoding, Charlie performs the $4$-outcome Bell measurement, $\mathrm{M}_{[4]}\equiv\{\mathrm{P}_{\phi^+},\mathrm{P}_{\phi^-},\mathrm{P}_{\psi^+},\mathrm{P}_{\psi^-}\}$. He, then, calculates $z_i=\mathbf{f}(x_i,y_i)$ for $\mathbf{f}\in\{\mbox{XOR},\mbox{XNOR}\}$ and computes the final outcome $\mathbb{F}(z_1,z_2)$. This suffices to compute all the nontrivial computations as in Theorem \ref{theorem1} in a delayed-choice manner.
\subsubsection{Proof of Corollary 3}\label{6.2.3}

Note that, it is clear from the above proof of Corollary \ref{coro2} that by performing the complete Bell measurement, Charlie is able to obtain the individual $z_{i}=f(x_{i}, y_{i})$, whenever the function $f$ is XOR or, XNOR. Therefore, following the same DCLC($2$) encoding protocol, for each two successive bits of their $n$-bit string, Alice and Bob will use a maximally entangled state and thus they require $n/2$-ebits for even $n$. Alternatively, For odd $n$, each of them requires $(n-1)/2$-ebits for first $(n-1)$-bits and $1$ product qubit for their last bit of information. After getting the $z_{i}$ values Charlie can evidently compute the given $\mathbb{F}$ and this completes the proof for both DCLC($n$) and DC$^{2}$LC($n$) whenever $f$ is balanced.
\subsection{Proof of Theorem 2}\label{6.3}
Since the operational dimension of any polygonal model $(\Omega_n,\mathcal{E}_n,\mathbb{T}_n)$ is $2$, both Alice and Bob are allowed to communicate one such system while performing a DCLC($2$) task. However, they can consider some composite models allowing entanglement. Now, for encoding, they will apply some reversible transformations on their respective part. Therefore, Remark \ref{rem4} leads us as follows : 
\begin{remark}\label{rem5}
	Encoded states, received by Charlie, are all either entangled or product.  
\end{remark}
We will now prove the Theorem \ref{theorem2} for {\bf Type-I} and {\bf Type-II} composite model. For odd-gon and even-gon, the proof will be discussed separately. \\\\ 
\underline{\it Odd-gon theory}\\\\
{\bf Type-I Models} : According to Remark \ref{rem5}, the encoded states, received by Charlie, are product states if Alice and Bob start their protocol with a product one. However, for such a scenario the encoding systems as well as the decoding devices (according the construction of \textit{Type-I} models) contain no entanglement. Therefore, by Proposition \ref{prop2} (in the main text) none of nontrivial DCLC tasks can be performed  in this case.\par
On the other hand, Alice and Bob may start the protocol with an entangled state. Then, Charlie receives all the encoded states as entangled, followed by Remark \ref{rem5}. A straightforward calculation shows that, $p(\bar{e}_i^A\otimes \bar{e}_j^B|\omega^{AB}_{kl}(p,q))\neq 0,~\forall i,j,k,l\in\{0,...,(n-1)\}~\&~\forall~\omega^{AB}_{kl}(p,q)$ in \eqref{e2} which leads to an unavoidable ambiguity while the decoding operation is performed by Charlie.\\\\
{\bf Type-II Models} : A decoding strategy with product effects ensures the equivalence with a classical strategy for a nontrivial DCLC($2$), and hence perfect accomplishment is impossible. Let us then move to entangled decoding strategies, and we consider, without loss of generality, Charlie's decoding measurement, $\mathcal{M}^{AB}\equiv\{e^{AB}_{00}(++),\bar{e}^{AB}_{00}(++)\}$. Suppose, Alice encodes her strings $\mathbf{x}$ and $\mathbf{x}^\prime~(\neq\mathbf{x})$ into the states $\omega_k^A$ and $\omega_l^A$ respectively, whereas the strings $\mathbf{y}$ and $\mathbf{y}^\prime~(\neq\mathbf{y})$, in Bob's side, are encoded by $\omega_s^B$ and $\omega_t^B$ respectively. For a particular encoded state, any of these two effects should get clicked perfectly in case of unambiguous decoding which leads to the following restrictions (see Table \ref{table1}).
\begin{center}
\begin{table}[h!]
	\begin{tabular}{|c||c|c| } 
		\hline
		~~Input strings~~ & ~~~$e^{AB}_{00}(++)$ clicks~~~ & ~~~$\bar{e}^{AB}_{00}(++)$ clicks~~~\\
		\hline\hline
		$\mathbf{x},\mathbf{y}$& $k=s$ & $k=s\oplus_n\frac{n\pm 1}{2}$\\ 
		\hline
		$\mathbf{x}^\prime,\mathbf{y}^\prime$& $l=t$ & $l=t\oplus_n\frac{n\pm 1}{2}$\\ 
		\hline
		$\mathbf{x},\mathbf{y}^\prime$& $k=t$ & $k=t\oplus_n\frac{n\pm 1}{2}$\\ 
		\hline
		$\mathbf{x}^\prime,\mathbf{y}$& $l=s$ & $l=s\oplus_n\frac{n\pm 1}{2}$ \\ 
		\hline
	\end{tabular}
	\caption{Conditions for which either of the entangled effects clicks sharply.}\label{table2}
\end{table}
\end{center}
It can be easily shown that $e^{AB}_{00}(++)$ and $\bar{e}^{AB}_{00}(++)$ will get clicked in $1:1$ ratio for any combination of these conditions. However, for a nontrivial computation, there should be at least a group $\mathcal{G}=\{\mathbf{x},\mathbf{x}^\prime,\mathbf{y},\mathbf{y}^\prime\}$ with $\mathbf{x}\neq\mathbf{x}^\prime~\&~\mathbf{y}\neq\mathbf{y}^\prime$ such that the string pairs are  mapped into $1:3$ ratio (Observation 1). Hence, no nontrivial computation can be performed perfectly in any Type-II {\it odd-gon} theory.\\

\underline{\it Even-gon theory}\\\\ 
{\bf Type-I Models:}
The effects, $e^A_i\otimes e^B_j$ and $\bar{e}^A_i\otimes\bar{e}^B_j$, can be clubbed together to get a single effect, $E_{i\otimes j}:=e^A_i\otimes e^B_j+\bar{e}^A_i\otimes\bar{e}^B_j$ since, $p(e^A_i\otimes e^B_j|\omega_{kl}^{AB}(p,q))=p(\bar{e}^A_i\otimes \bar{e}^B_j|\omega_{kl}^{AB}(p,q),~\forall~i,j,k,l,p,q$. Similarly, we have, $\bar{E}_{i\otimes j}:=e^A_i\otimes\bar{e}^B_j+\bar{e}^A_i\otimes e^B_j$. Clubbing the effects in a different manner will do nothing but increase the ambiguity. Consider that Alice and Bob start the protocol with $\omega_{00}^{AB}(++)$ and Charlie performs the decoding measurement, $\mathcal{M}^{A\otimes B}\equiv\{E_{0\otimes 0},\bar{E}_{0\otimes 0}\}$ without loss of generality. Encoding of different bit-strings can be accomplished by applying the proper reversible transformations on $\omega_{00}^{AB}(++)$. The probabilities to obtain the effect $E_{0\otimes0}$ on  different entangled states are given by,
\begin{itemize}
	\item [i)]$p\left(E_{0\otimes0}|\omega_{kl}^{AB}(++)\right)=\frac{1}{2}\left[1+r^2_n\cos\left(\frac{\pi}{n}-\frac{2\pi}{n}(k-l)\right)\right]$,
	\item [ii)]$p\left(E_{0\otimes0}|\omega_{kl}^{AB}(--)\right)=\frac{1}{2}\left[[1+r^2_n\cos\left(\frac{\pi}{n}+\frac{2\pi}{n}(k-l)\right)\right]$,
	\item [iii)]$p\left(E_{0\otimes0}|\omega_{kl}^{AB}(+-)\right)=\frac{1}{2}\left[[1+r^2_n\cos\left(\frac{\pi}{n}+\frac{2\pi}{n}(k+l)\right)\right]$,
	\item [iv)]$p\left(E_{0\otimes0}|\omega_{kl}^{AB}(-+)\right)=\frac{1}{2}\left[[1+r^2_n\cos\left(\frac{3\pi}{n}+\frac{2\pi}{n}(k+l)\right)\right]$,
\end{itemize}
where, $\omega_{kl}^{AB}(p,q)=\left(\mathcal{T}_k^p\otimes\mathcal{T}_l^q\right)\omega_{00}^{AB}(++)$. To avoid the ambiguity, we have to choose the proper reversible transformations maintaining the restrictions listed in Table \ref{table3}.
\begin{table}[h!]
	\resizebox{16cm}{!} {
		\begin{tabular}{ c c c c c|l }
			\cline{2-5}
			& \multicolumn{4}{|c|}{\multirow{2}{*}{Restrictions on $k$ and $l$ when the states are}} \\
			& \multicolumn{1}{|c}{}& & & &\\ 
			\cline{2-5}
			& \multicolumn{1}{|c}{\multirow{2}{*}{$\omega_{kl}^{AB}(++)$}} & \multicolumn{1}{|c}{\multirow{2}{*}{$\omega_{kl}^{AB}(--)$}} & \multicolumn{1}{|c}{\multirow{2}{*}{$\omega_{kl}^{AB}(+-)$}} & \multicolumn{1}{|c|}{\multirow{2}{*}{$\omega_{kl}^{AB}(-+)$}} \\ 
			& \multicolumn{1}{|c}{} & \multicolumn{1}{|c}{} & \multicolumn{1}{|c}{} & \multicolumn{1}{|c|}{} &\\ 
			\cline{1-5}
			\multicolumn{1}{|c}{\multirow{2}{*}{$p[E_{0\otimes0}|\omega_{kl}^{s_As_B}]=1$}} & \multicolumn{1}{|c}{\multirow{2}{*}{$k=l$, $k=l\oplus_n 1 $}} & \multicolumn{1}{|c}{\multirow{2}{*}{$k=l$, $k=l\oplus_n(n-1)$}} & \multicolumn{1}{|c}{\multirow{2}{*}{$k=-l\oplus_n n$, $k=-l\oplus_n(n-1)$}} & \multicolumn{1}{|c|}{\multirow{2}{*}{$k=-l\oplus_n(n-1)$, $k=-l\oplus_n 2(n-1)$}}\\
			\multicolumn{1}{|c}{} & \multicolumn{1}{|c}{} & \multicolumn{1}{|c}{} & \multicolumn{1}{|c}{} & \multicolumn{1}{|c|}{} &\\
			\cline{1-5}
			\multicolumn{1}{|c}{\multirow{2}{*}{$p[\bar{E}_{0\otimes0}|\omega_{kl}^{s_As_B}]=1$}} & \multicolumn{1}{|c}{\multirow{2}{*}{$k=l\oplus_n \frac{n}{2},k=l\oplus_n(\frac{n}{2}+1)$}} & \multicolumn{1}{|c}{\multirow{2}{*}{$k=l\oplus_n \frac{n}{2},k=l\oplus_n(\frac{n}{2}-1)$}} & \multicolumn{1}{|c}{\multirow{2}{*}{$k=-l\oplus_n \frac{n}{2},k=-l\oplus_n(\frac{n}{2}-1)$}} & \multicolumn{1}{|c|}{\multirow{2}{*}{$k=-l\oplus_n(\frac{n}{2}-1)$, $k=-l\oplus_n(\frac{n}{2}-2)$}}\\
			\multicolumn{1}{|c}{} & \multicolumn{1}{|c}{} & \multicolumn{1}{|c}{} & \multicolumn{1}{|c}{} & \multicolumn{1}{|c|}{} &\\
			\cline{1-5}
	\end{tabular}}
	\caption{The allowed integer values of $k$ and $l$ have been depicted for which either of the entangled effects clicks sharply.}\label{table3}
\end{table}
Suppose, Alice applies $\mathcal{T}_k^p$ \& $\mathcal{T}_{k^\prime}^{p^\prime}$ when she receives the strings $\mathbf{x}$ \& $\mathbf{x}^\prime~(\neq\mathbf{x})$ respectively, and Bob applies $\mathcal{T}_l^q$ \& $\mathcal{T}_{l^\prime}^{q^\prime}$ for the strings $\mathbf{y}$ \& $\mathbf{y}^\prime~(\neq\mathbf{y})$ similarly on their shared state $\omega_{00}^{AB}(++)$. Compared to the {\it odd-gon} theories, there are more possibilities for encoding in {\it even-gon} cases as the number of entangled states are more. We consider a particular case with $p=+$, $p^\prime=-$, $q=+$, and $q^\prime=-$. With the help of the Table \ref{table3}, we arrive at the following conditions (see Table \ref{table4}) for unambiguous decoding.
\begin{table}[h!]
	\begin{tabular}{|c||c|c| } 
		\hline
		~~Input strings~~ & ~~~~~Encoded State~~~~~ & ~~~~~~~~~~Conditions~~~~~~~~~\\
		\hline\hline
		\multirow{3}{4em}{~~~$\mathbf{x},\mathbf{y}$}& \multirow{3}{4em}{$\omega_{kl}^{AB}(++)$} & $k=l~\&~k=l\oplus_n1$\\ 
		&  & or\\ 
		&  & $k=l\oplus_n\frac{n}{2}~\&~k=l\oplus_n(\frac{n}{2}+1)$\\
		\hline
		\multirow{3}{4em}{~~~$\mathbf{x}^\prime,\mathbf{y}^\prime$}& \multirow{3}{4em}{$\omega_{k^\prime l^\prime}^{AB}(--)$} & $k^\prime=l^\prime~\&~k^\prime=l^\prime\oplus_n(n-1)$\\ 
		&  & or\\
		&  & $k^\prime=l^\prime\oplus_n\frac{n}{2}~\&~k^\prime=l^\prime\oplus_n(\frac{n}{2}-1)$\\ 
		\hline
		\multirow{3}{4em}{~~~$\mathbf{x},\mathbf{y}^\prime$}& \multirow{3}{4em}{$\omega_{k l^\prime}^{AB}(+-)$} & $k=-l^\prime\oplus_n n~\&~k=-l^\prime\oplus_n(n-1)$\\ 
		&  & or\\
		&  & $k=-l^\prime\oplus_n\frac{n}{2}~\&~k=-l^\prime\oplus_n(\frac{n}{2}-1)$\\ 
		\hline
		\multirow{3}{4em}{~~~$\mathbf{x}^\prime,\mathbf{y}$}& \multirow{3}{4em}{$\omega_{k^\prime l}^{AB}(-+)$} & $k^\prime=-l\oplus_n(n-1)~\&~k^\prime=-l\oplus_n2(n-1)$\\ 
		&  & or\\
		&  & ~~~$k^\prime=-l\oplus_n(\frac{n}{2}-1)~\&~k^\prime=-l\oplus_n(\frac{n}{2}-2)$~~~\\ 
		\hline
		\hline
	\end{tabular}
	\caption{Conditions for unambiguous decoding.}\label{table4}
\end{table}
These conditions lead to the fact that the effects $E_{0\otimes0}$ and $\bar{E}_{0\otimes0}$ will get clicked either in $1:1$ or in $1:0$ ratio resulting a trivial computation. In a similar fashion, one can argue the same for any choice of $p,p^\prime,q,q^\prime\in\{+,-\}$. Hence, no nontrivial computations can be executed by this type of theories.\\\\
{\bf Type-II Models:} In this case, the arguments go as in the case of {\bf Type-II} {\it odd-gon} models and it turns out that the decoding effects get clicked in $1:1$ ratio. Hence, all the computations, which can be accomplished, are trivial by this kind of theories. This completes proof of the Theorem \ref{theorem2}. 

\end{document}